%% file: main.tex
\documentclass[letterpaper, 10 pt, conference]{ieeeconf}
\IEEEoverridecommandlockouts
\usepackage{cite}
\usepackage{url}
\usepackage{array}
\usepackage{xifthen}
\usepackage{graphicx} 
\usepackage{amsmath, amssymb, amsfonts}
\usepackage{algorithm}
\usepackage{algpseudocode}
\usepackage{caption}
\usepackage{tablefootnote}
\usepackage[table]{xcolor}
\usepackage{diagbox}
\usepackage[strict]{changepage}
\usepackage{framed}
\usepackage{tikz}
\usetikzlibrary{calc,intersections}
\usetikzlibrary{fit,backgrounds,positioning,decorations.pathreplacing, calligraphy, matrix}
\newcolumntype{C}[1]{>{\centering\arraybackslash}p{#1}}
\newcommand\Opt[4]{
  \ifthenelse{\isempty{#2}}%
  {\mbox{{#1}}}
  {\underset{{#2}}{\mbox{{#1}}}}
  \;
  \ifthenelse{\isempty{#4}}%
  {{#3}}
  {\Big\{ {#3}\  \Big| \  {#4} \Big\}}
}
\newcommand{\norm}[1]{\left\lVert#1\right\rVert}
\newcommand{\Plus}{\mathord{\begin{tikzpicture}[baseline=0ex, line width=1.5, scale=0.13]
\draw (1,0) -- (1,2);
\draw (0,1) -- (2,1);
\end{tikzpicture}}}
\newcommand{\Minus}{\mathord{\begin{tikzpicture}[baseline=0ex, line width=1.5, scale=0.1]
\draw (0,0.8) -- (2,0.8);
\end{tikzpicture}}}
\newtheorem{definition}{Definition}[section]
\newtheorem{proposition}{Proposition}[section]

\definecolor{formalshade}{rgb}{1,1,1}

\title{\LARGE \bf Homomorphically encrypted gradient descent algorithms \\ for quadratic programming}

\author{\normalsize Andr\'e Bertolace, Konstantinos Gatsis, Kostas
  Margellos\thanks{The authors are with the Department of Engineering
    Science, University of Oxford, Oxford OX1 3PJ, U.K.  E-mail:
    kostas.margellos@eng.ox.ac.uk, konstantinos.gatsis@eng.ox.ac.uk,
    andre.bertolace@eng.ox.ac.uk (corresponding author) }
  \thanks{\copyright 2023 IEEE.  Personal use of this material is
    permitted.  Permission from IEEE must be obtained for all other
    uses, in any current or future media, including
    reprinting/republishing this material for advertising or
    promotional purposes, creating new collective works, for resale or
    redistribution to servers or lists, or reuse of any copyrighted
    component of this work in other works.}}

\date{\today} 

\begin{document}

\maketitle 

\begin{abstract}
In this paper, we evaluate the different fully homomorphic encryption schemes, propose an implementation, and numerically analyze the applicability of gradient descent algorithms to solve quadratic programming in a homomorphic encryption setup. The limit on the multiplication depth of homomorphic encryption circuits is a major challenge for iterative procedures such as gradient descent algorithms. Our analysis not only quantifies these limitations on prototype examples, thus serving as a benchmark for future investigations, but also highlights additional trade-offs like the ones pertaining the choice of gradient descent or accelerated gradient descent methods, opening the road for the use of homomorphic encryption techniques in iterative procedures widely used in optimization based control. In addition, we argue that, among the available homomorphic encryption schemes, the one adopted in this work, namely CKKS, is the only suitable scheme for implementing gradient descent algorithms. The choice of the appropriate step size is crucial to the convergence of the procedure. The paper shows firsthand the feasibility of homomorphically encrypted gradient descent algorithms.
\end{abstract}


\IEEEpeerreviewmaketitle

\input{chapters/introduction.tex}
\input{chapters/gd.tex}
\input{chapters/he.tex}
\input{chapters/he_gd.tex}
\input{chapters/numerical.tex}
\input{chapters/conclusion.tex}


\bibliographystyle{plain}
\bibliography{main.bib}


\end{document}

%% file: chapters/introduction.tex
\section{Introduction}

Homomorphic encryption (HE) is a ground-breaking mathematical method that enables the analysis or manipulation of encrypted data without revealing its content \cite{Gentry10}. In doing so, HE permits the secure delegation of data processing to third-party cloud providers. Several encryption schemes, such as Paillier \cite{Paillier99} or El Gamal \cite{ElGamal85} are partially HE schemes\footnote{Partially HE schemes enable the implementation of   either addition or multiplication on encrypted data, but not both,  whereas fully HE schemes enable the implementation of both addition and multiplication operations.}. In 2009, Gentry \cite{Gentry09} proposed the first fully HE scheme$^{\text{1}}$. The computational overhead of the scheme was significant, but it showed that such schemes are indeed possible. Since then these approaches have been further developed. Currently, the state of the art schemes are BFV \cite{BFV12a, BFV12b}, CKKS \cite{CKKS17}, BGV \cite{BGV12}, and GSW \cite{GSW13}. The computational overhead remains large but it has been brought down to a level where these schemes can be implemented in practice. 

In applications of control and decision-making, the benefits of delegating data processing without giving away access to the data are tremendous. The use of HE schemes applied in control theory is at its infancy, however, encrypted linear controllers have been implemented. Most results use partially HE schemes \cite{Alexandru21a, Alexandru21b, Darup19, Cheon18, Farokhi16, Farokhi17}, but approaches using fully HE schemes also exist \cite{Kim16, Darup18, Kosieradzki22}. In addition, \cite{Darup21} provides a detailed overview of the current status of research in the encrypted control for networked systems and \cite{Kim22} a comparison of different encrypted control approaches. 

Yet, the implementation of algorithms in a HE setup is far from trivial. For instance, many HE schemes use random noise to guarantee the security of the encrypted data. This noise compounds at every arithmetic operation, resulting in a limited number of sequential arithmetic operations performed by an encryption circuit, in special, multiplication operations.

In this paper, we would like to understand the limits imposed by HE computation on challenging computation tasks beyond the controller implementation problems studied in the literature. Specifically, we consider the problem of solving quadratic programming (QP) problems. QP is commonly used in several control problems like those arising in state estimation under minimum square error and model predictive control. Numerically solving such a task often requires \textit{iterative} methods (gradient descent) and the limit on the multiplication depth of HE circuits is a major challenge for iterative procedures. As a result, given the HE multiplication depth limits, we would like to determine the most appropriate iterative methods for QP. In our case, we adapt and implement gradient descent (GD) and accelerated gradient descent (AGD) algorithms to solve a QP in a HE manner. Our contributions are threefold:
\begin{enumerate}
\item We argue that among the available HE schemes CKKS is the only scheme suitable to handle GD and AGD iterations in a HE setup as it allows handling real-valued operations, an important feature, especially in the selection of an appropriate step-size that ensures the convergence of the underlying algorithm. 
\item We implement our own HE matrix multiplication algorithm that is more efficient, in terms of multiplication depth, than other algorithms proposed in the literature \cite{JKLS18}.
\item We demonstrate that in the HE setup the condition number of the matrix of the quadratic term of the objective function plays an important role in determining which algorithm is preferred. The encrypted version of AGD is the preferred algorithm only for matrices with higher condition numbers. This is in contrast to plain-text optimization, where AGD is typically preferable due to its superior convergence rate. The reason being once an encryption circuit is defined, the number of sequential arithmetic operations is fixed but at each step, AGD performs one extra multiplication when compared to GD. These extra iterations pays off, specially for matrices with lower conditional number.
\end{enumerate}

Other works in the literature, \cite{Alexandru21a}, also proposed to solve QPs in a secure/distributed manner. We differ from this work by using fully homomorphic encryption instead of partial homomorphic encryption schemes. Furthermore, \cite{Darup18}, presented an encrypted model predictive control scheme for linear constrained systems, also using partial homomorphic encryption schemes and not solving the QP but instead using the corresponding piece-wise affine control law if explicitly given.

It should be mentioned that our paper focuses on unconstrained quadratic problems due to inherent limitations in operations that are allowed to be performed by available HE schemes. However, our analysis not only quantifies these limitations on prototype examples thus serving as a benchmark for future investigations, but also highlights additional trade-offs like the ones pertaining the choice of GD or AGD methods, opening the road for the use of HE methods in iterative methods widely used in optimization based control.

Finally, to emphasize that this line of research is still at an early stage, we note that our implementation using state of the art HE tools (Microsoft SEAL) permits just a modest amount of gradient steps, beyond what would be required in practical applications.

The rest of the paper is organized as follows: Section \ref{sec:SecII} provides some background information on gradient methods, while Section \ref{sec:SecIII} discusses our choice for an HE scheme. Section \ref{sec:SecIV} discusses the implementation of the suggested scheme and provides an extension to the matrix multiplication operation. Section \ref{sec:SecV} provides a detailed numerical study on the application of GD and AGD algorithms for QPs, while Section \ref{sec:SecVI} concludes the paper.

%% file: chapters/gd.tex
\section{Descent algorithms for unconstrained quadratic programming} \label{sec:SecII}

QP has been a very successful tool for modeling many real-life problems. It is extensively used in applications that involve the variance minimization, such as in the formulation of portfolio optimization problems or in solving the ordinary least square (OLS) problem. In fact, many problems in physics on engineering can be formulated as some form of energy minimization problem, in which the energy can simply be formulated as a quadratic form,

\[
\Opt{min}{x \in \mathbb{R}^n}{\frac{1}{2} x^T Q x + p^T x}{},
\]
where $Q \in \mathbb{R}^{n \times n}$ and $p \in \mathbb{R}^n$. Note that the QP is convex if $Q \succcurlyeq 0$.

This unconstrained QP has a closed form solution, it requires however the inversion of a matrix, a procedure that involves other operations than additions and multiplications, posing hence a challenge for its implementation in a HE setup. An alternative solution is to use gradient descent methods to solve the QP problem. This is a class of iterative algorithms that provide a simple way \cite{Bubeck15} to minimize a differentiable function $f$, 
\[
\Opt{min}{x \in \mathbb{R}^n}{f(x)}{}.
\]

Starting at an initial estimate, it iteratively updates 
\[
x_{t+1} = x_t - \eta \nabla f (x_t),
\]

where $\nabla f (x_t)$ denotes the gradient of $f$ calculated at $x_t$ and $\eta$ the step-size, until reaching a desired tolerance in the solution. Particularly to the QP case, the gradient takes a linear form, $\nabla f(x) = Q x + p$.

Methods of this type have a convergence rate which is independent of the dimension $n$ of the solution space. This feature makes them particularly attractive for optimization in very high dimensions \cite{Bubeck15}. The convergence is however deeply linked to the step-size $\eta$, be it too small, the algorithm may take too long to converge, be it too high, it may diverge. 

Properties such as smoothness or strong convexity of the objective function $f$ do play a relevant role in choosing $\eta$ and a variant of the algorithm with faster convergence.

\begin{definition}
  A continuous differentiable function $f$ is $\beta$-smooth if the gradient $\nabla f$ is $\beta$-Lipschitz, i.e., \[
  \norm{\nabla f(x) - \nabla f(y)} \le \beta \norm{x-y}, ~\forall x,y \in \mathbb{R}^n.
  \]
\end{definition}

\begin{definition}
  A function $f$ is $\alpha$-strongly convex, with $\alpha > 0$, if for any $x$, $y$ it satisfies the following sub-gradient inequality, i.e., 
  \[
  f(x) - f(y) \le \nabla f(x)^T (x-y) - \frac{\alpha}{2} \norm{x-y}^2,~ \forall x,y \in \mathbb{R}^n.
  \]
\end{definition}

 Given these definitions, an immediate consequence is that if $f$ is twice differentiable, then $f$ is $\alpha$-strongly convex if the eigenvalues of the Hessian of $f$ are larger than or equal to $\alpha$. A quadratic $f$, as in our QP, is $\lambda_{\max}$-smooth and $\lambda_{\min}$-strongly convex, where $\lambda_{\max}, \lambda_{\min} > 0$ are, respectively, the maximum and minimum eigenvalues of the matrix $Q$. The ratio $\kappa=\frac{\lambda_{\max}}{\lambda_{\min}}$ is the condition number of the matrix $Q$ and it plays an important role in the convergence of descent algorithms. For reference, both Nesterov's accelerated gradient descent (AGD) and Gradient Descent (GD) methods ~\cite{Bubeck15} for the smooth and strongly convex quadratic function converge exponentially fast (see Table \ref{tbl:algoComparison}):

\begin{table}[h!]
  \centering
\begin{tabular}{C{1.25cm}|C{1.5cm}|C{2cm}|C{2cm}}
  $f$ & Algorithm & Convergence rate & Iterations needed \\ 
  \hline
  $\alpha$-conv., $\beta$-smooth & AGD & $R^2 \exp \biggl( - \frac{t}{\sqrt{\kappa}} \biggr) $ & $\sqrt{\kappa} \log \biggl( \frac{R^2}{\epsilon} \biggr)$ \\
  \hline
  $\alpha$-conv., $\beta$-smooth & GD & $R^2 \exp \biggl( - \frac{t}{\kappa} \biggr) $ & $\kappa \log \biggl( \frac{R^2}{\epsilon} \biggr)$
\end{tabular}
\caption{Convergence rate and other parameters for different algorithms. $R= \norm{x_1 - x^*}$ is the distance from the initial estimate to the optimal value. The number of iterations is directly derived from the convergence rate for a fixed tolerance (in terms of distance from optimal value) $\epsilon$.}
\label{tbl:algoComparison}
\end{table}


%% file: chapters/he.tex
\section{Homomorphically encrypted arithmetic} \label{sec:SecIII}

\subsection{Homomorphic encryption schemes}

HE schemes have been developed using different approaches. BFV and BGV perform operations modulo integer whereas CKKS implements approximated fixed point arithmetics. The security of these schemes is based on the Ring Learning With Errors (RLWE) problem, a variant of the Learning With Errors problem (LWE), in which the goal is to distinguish random linear equations, which have been perturbed by a small amount of noise from uniform ones \cite{LWE12}. The HE cipher is defined by a pair $E, D$, of encryption-decryption algorithms respectively. $E$ takes a public key $pk$ along with a message $m$ as inputs and outputs a cipher-text $\mathtt{c}$, as $\mathtt{c} = E(pk,m)$. The decryption algorithm, $D$, takes a secret key $sk$ along with the cipher-text $\mathtt{c}$ as inputs and outputs the message $m = D(sk, \mathtt{c})$. The algorithms are parameterized by a security parameter $\lambda$ which plays a direct role in the derivation of the $sk$. In addition, these schemes exploit the structure of polynomial rings for its plain-text and cipher-text spaces, the cyclonomic polynomial, $R[\mathbb{Z}_q] = \mathbb{Z}_q[X]/(X^N + 1)$. All schemes make use of random  variables with values sampled from a discrete Gaussian distribution with a pre-defined variance and random variables sampled from a ternary distribution $\{-1, 0, 1\}$ \cite{CKKS17}. 

%
\textbf{BFV:} In the BFV scheme, \cite{BFV12a, BFV12b}, the plain-text and cipher-text spaces are defined by two distinct rings, $R[\mathbb{Z}_t]$ and $R[\mathbb{Z}_q]$, where $t$ and $q$ are parameters of the plain-text and cipher-text coefficients, respectively. 

\textbf{BGV:} The BGV scheme, \cite{BGV12}, is similar to the BFV. The plain-text and cipher-text spaces are defined by two distinct rings, $R[\mathbb{Z}_t]$ and $R[\mathbb{Z}_q]$. 

\textbf{CKKS:} The CKKS scheme, \cite{CKKS17}, is often quoted as being the most efficient method to perform approximate HE computations over real and complex numbers \cite{Kim20}. It can be considered as a noisy channel \cite{Lee20}, a simple encryption/decryption procedure adds noise to the original message. The scheme exploits the structure of integer polynomial rings for its plain-text and cipher-text spaces, $R[\mathbb{Z}_q]$ and $R[\mathbb{C}]$. This polynomial ring is combined with a canonical embedding transformation $\sigma : \mathcal{S}\rightarrow \mathbb{C}^N$ that encodes/decodes a vector\footnote{The space size is actually $N/2$ because the roots of the cyclonomic polynomial lie on the unit circle and are pairwise complex conjugate.} in $\mathbb{C}^N$ to/from the ring of cyclonomical complex polynomials $R[\mathbb{C}]$. To encode a message $x \in R[\mathbb{C}]$ one applies the inverse embedding transformation to get $\mu=\sigma^{-1}(x) \in R[\mathbb{C}]$, then scale $\mu$ by a factor $\Delta = 2^p$ and round to obtain the plain-text $m=\lfloor \Delta \cdot \mu \rceil \in R[\mathbb{Z}_q]$.

%
%
%

\subsection{Scheme choice}
We claim that the most suitable choice for HE versions of the GD algorithm (similar considerations hold for the AGD one) is the CKKS scheme. The main reason for such claim is related to the selection of the step-size $\eta$. Note that $\eta$ should be sufficiently small for GD to converge. This is summarized in the following proposition; it is a standard result but we present a proof below for completeness.

\begin{proposition}
  Consider a QP with $Q \succcurlyeq 0$ with $Q=Q^\top$, and let $\lambda_{\max}$ denote the maximum eigenvalue of $Q$. The GD method converges for any $\eta < \frac{2}{\lambda_{\max}}$.
\end{proposition}

\begin{proof}
  Given $f(x) = \frac{1}{2} x^T Q x + p^T x$, we have that $\nabla f(x) = Q x + p$ and the iterative GD procedure takes the form $x_{t+1} = (I - \eta Q) x_t - \eta p$. Let $x^*$ be an unconstrained minimizer of $f$. As such, $\nabla f(x^*) = Q x^* + p = 0$, which in turn implies that $x^* = (I - \eta Q) x^* - \eta p$. We thus have that $x_{t+1} - x^* = (I - \eta Q) (x_t - x^*)$ and consequently $\norm{x_{t+1} - x^*} \leq \norm{I - \eta Q} \norm{(x_t - x^*)} \leq \norm{I - \eta Q}^{t+1} \norm{(x_0 - x^*)}$. The latter implies that $\lim_{t \to \infty} \norm{x_{t+1} - x^*} = 0$ if the maximum eigenvalue of $(I - \eta Q)$ is less than $1$, which can be achieved if $\eta < \frac{2}{\lambda_{\max}}$.
\end{proof}

A direct consequence of this fact is that if using BGV or BFV that require integer step-sizes, one can only ensure convergence for matrices with $\lambda_{\max} < 2$ that is the only choice that allows for an integer step-size $\eta$. The minimum then value of such step-size would be $\eta =1$, which in turn may lead to an erratic numerical behaviour. Additionally, to be able to use BFV or BGV, one would need to limit the calculations to integer matrices $Q \in \mathbb{Z}^{n \times n}$, or manipulate $Q \in \mathbb{R}^{n \times n}$ to be made integer. Towards this direction, \cite{Kim22, Kim23}, suggest the following manipulations:
\begin{itemize}
    \item "Scaling-up" the real numbers by a factor, say $10^8$, replicating a fixed point arithmetic, and proceed by calculating using the given integer numbers.
    \item Converting the matrix $Q$ by finding an invertible matrix $T \in \mathbb{R}^{n \times n} $ such that $T Q T^{-1} \in \mathbb{Z}^{n \times n}$.
\end{itemize}

The former is not a practical solution as the result of multiple multiplications will overflow and the output after the decryption will be incorrect \cite{Kim22}, whereas the latter implies limiting ourselves to matrices $Q$ in which every eigenvalue has an integer real and imaginary part \cite{Kim22, Kim23}.
In summary, BGV and BFV are only suitable schemes for integer matrices or matrices that have integer eigenvalues, which for our setting would require $\lambda_{\max} <2$. This would imply working only with identity matrices, $Q=I$, if we working with integer matrices which are symmetric and positive definite.  As such, for the purpose of an iterative methodology like GD and AGD, CKKS is preferable. 

%% file: chapters/he_gd.tex
\section{Homomorphically encrypted gradient descent algorithms} \label{sec:SecIV}
\subsection{Algorithm description}

For the HE version of gradient descent methods, let us start by defining the following arithmetic operators:
\begin{itemize}
\item $\Plus$/$\Minus$: the addition/subtraction of two cipher-texts;
\item $\bullet$: the multiplication of two cipher-texts;
\item $\odot$: the multiplication of a plain-text and a cipher-text;
\end{itemize}

Let us further assume that the user calculates $\lambda_{\min}$, $\lambda_{\max}$, and sends these as plain-text, i.e. not encrypted, values to the solver. Together with these constants, the user also sends the encrypted matrix and vector, $\mathtt{Q}=E(pk, Q)$ and $\mathtt{p} = E(pk, p)$ that determine the QP. The encrypted version of the descent algorithms will still proceed in an iterative fashion. The only difference is that one would be iterating over cipher-texts $\mathtt{c_t}$ instead of plain-text $x_t$. When iterating over cipher-texts, two steps deserve special attention, the stopping criteria $\left|\mathtt{c_{t+1}} \Minus \mathtt{c_t}\right| > \epsilon$ and the matrix multiplication procedure, referred to $\textsc{MMult}$ (Algorithm \ref{algo:MMult}) and discussed in the sequel. The latter is relevant because of the exponential growth of the noise level with the multiplication depth. The authors in \cite{JKLS18} do propose a fully-homomorphic matrix multiplication algorithm to perform these operations in a more efficient manner. On the other hand, for the stopping rule, determining whether an encrypted value is larger than another encrypted value or even a plain-text without decrypting both values is directly not feasible, but complex approaches to implement comparisons have appeared in \cite{comparisonBFV} for BFV and \cite{comparisonCKKS} for CKKS.

Given the challenge to implement the stopping rule in an HE setup, we propose that the HE version of the AGD algorithm is slightly modified:
\begin{itemize}
\item Instead of specifying the tolerance $\epsilon$, the user fixes the number of iterations $N$. Given $N$ and the convergence rate of the algorithms we can infer the attainable tolerance (Table \ref{tbl:algoComparison}).
\item The user may hand in the initial estimate $\mathtt{x_0}$, although this is not necessary.
\end{itemize}

The HE versions of the the AGD and GD still follow an iterative procedure. These take the form of (Algorithm \ref{algo:HEAGDQP}) and (Algorithm \ref{algo:HEGDQP}) respectively and are very similar to the usual AGD and GD algorithms. The main difference is the use of HE arithmetic operators and the special $\textsc{MMult}$ matrix multiplication procedure.

\begin{algorithm}[h!]
\caption{HE AGD for an unconstrained QP}
\label{algo:HEAGDQP}
\begin{algorithmic}[1]
  \Function{HEagdQP}{$\mathtt{Q}, \mathtt{p}, d, \lambda_{\min}, \lambda_{\max}, \mathtt{x_0}, N$}
  \State $\kappa \gets \frac{\lambda_{\max}}{\lambda_{\min}}$
  \State $\mathtt{x_-} \gets \mathtt{x_0}$
  \State $\mathtt{y_-} \gets \mathtt{x_0}$
  \State $\eta \gets  \frac{-1}{\lambda_{\max}}$
  \For {$t = 0$ to $N-1$}
  \State $\mathtt{y_+} \gets \mathtt{x_-} \Plus \textsc{MMult}(\mathtt{Q}, \mathtt{x_-}, d, \eta) \Plus \eta \odot \mathtt{p}$
  \State $\mathtt{x_+} \gets \biggl( 1 + \frac{\sqrt{\kappa} -1}{\sqrt{\kappa} + 1} \biggr)  \odot \mathtt{y_+} \Minus  \frac{\sqrt{\kappa} -1}{\sqrt{\kappa} +1} \odot \mathtt{y_-}$
  \State $\textsc{Relinearize}(\mathtt{x_+})$
  \State $\mathtt{y_-} \gets \mathtt{y_+}$
  \State $\mathtt{x_-} \gets \mathtt{x_+}$
  \EndFor
  \State \textbf{return} $\mathtt{x_+}$
  \EndFunction
\end{algorithmic}
\end{algorithm}

\begin{algorithm}[h!]
\caption{HE GD for an unconstrained QP}
\label{algo:HEGDQP}
\begin{algorithmic}[1]
  \Function{HEgdQP}{$\mathtt{Q}, \mathtt{p}, d, \lambda_{\min}, \lambda_{\max},
    \mathtt{x_0}, N$}
  \State $\kappa \gets \frac{\lambda_{\max}}{\lambda_{\min}}$
  \State $\mathtt{x_-} \gets \mathtt{x_0}$
  \State $\eta \gets  \frac{-2}{\lambda_{\min}+ \lambda_{\max}}$
  \For
  {$t = 0$ to $N-1$}
  \State
  $\mathtt{x_+} \gets \mathtt{x_-} \Plus \textsc{MMult}(\mathtt{Q},
  \mathtt{x_-}, d, \eta) \Plus \eta \odot \mathtt{p}$
  \State
  $\mathtt{x_+} \gets \mathtt{x_-}$
  \EndFor
  \State \textbf{return}
  $\mathtt{x_+}$
  \EndFunction
\end{algorithmic}
\end{algorithm}

The only difference between the two algorithms is the presence of the extra two $\odot$ and one $\Minus$ operations on line 8 of the (Algorithm \ref{algo:HEAGDQP}) when compared to (Algorithm \ref{algo:HEGDQP}). These operations are intrinsic to the accelerated gradient method as the method uses additional past information to update to the next step.

\subsection{Matrix multiplication seen differently}

Because we would like to study gradient descent methods, let us consider a simple matrix multiplication. Halevi and Shoup \cite{Halevi18} introduced a method (Algorithm \ref{algo:HaleviShoup}) to evaluate an arbitrary linear transformation on encrypted vectors. They exploit the diagonal encoding of a matrix to easily express the matrix-vector multiplication by combining rotation and constant multiplication operations

\begin{algorithm}[h!]
\caption{Halevi-Shoup \textsc{LinTrans} algorithm}
\label{algo:HaleviShoup}
\begin{algorithmic}
  \Function{LinTrans}{$\mathtt{c}, U$}
  \State $n \gets \dim(U)$
  \State $\mathtt{cU} \gets \mathtt{c} \odot u_0$
  \For {$l = 1$ to $n-1$}
    \State $\mathtt{cU} \gets \mathtt{cU} \Plus \textsc{Rot}(\mathtt{c}, l) \odot u_l$
  \EndFor
  \State $\textsc{Relinearize}(\mathtt{cU})$
  \State \textbf{return} $\mathtt{cU}$
  \EndFunction
\end{algorithmic}
\end{algorithm}

In \cite{JKLS18} the authors elaborated further on the Halevi-Shoup method and proposed a new matrix multiplication scheme (JKLS) that allows for a ciphered-matrix multiplication that uses only one cipher-text per matrix following a row-ordering encoding $A \rightarrow a$. Although convenient, the JKLS algorithm needs 2 $\odot$ operations \cite{JKLS18}.


We propose a modified version of a matrix multiplication algorithm with 1 less $\odot$ multiplication step:
\begin{align*}
  a_k &= V_k \odot a \ , b_k = W_k \odot b \ , \ k=0,\ldots,d-1 \\
  ab &= \sum_{k=0}^{d-1} a_k \bullet b_k
\end{align*}

with:
\begin{equation*}
  V_k (d \cdot i + j, l)= \begin{cases}
      1 & \text{if } l= d \cdot i + [i + j + k]_d\\
      0 & \text{otherwise}\\
    \end{cases} 
\end{equation*} 
\begin{equation*}
  W_k(d \cdot i + j, l)= \begin{cases}
      1, & \text{if } l = d \cdot [i + j + k]_d + j\\
      0, & \text{otherwise}\\
    \end{cases} 
\end{equation*}

where $[\cdot]_d$ is a shortcut for $\cdot$ modulo $d$.
\begin{figure*}[ht]
  \begin{center}
    \resizebox{0.75\textwidth}{!}{
    \begin{tikzpicture}[every matrix/.style={matrix of nodes,
        nodes={draw, minimum size=1.5mm, anchor=center,
          inner sep=0pt, line width=0.5pt}, 
        column sep=-0.5pt, row sep=-0.5pt, font=\scriptsize}]

\matrix [ampersand replacement=\&] (V0) { 
|[fill=blue!50]| \& $ $ \& $ $ \& $ $ \& \& \& \& \& \& \& \& \& \& \& \& \\ 
$ $ \& |[fill=blue!50]| \& $ $ \& $ $ \& \& \& \& \& \& \& \& \& \& \& \& \\ 
$ $ \& $ $ \& |[fill=blue!50]| \& $ $ \& \& \& \& \& \& \& \& \& \& \& \& \\ 
$ $ \& $ $ \& $ $ \& |[fill=blue!50]| \& \& \& \& \& \& \& \& \& \& \& \& \\ 
\& \& \& \& $ $ \& |[fill=gray]| \& $ $ \& $ $ \& \& \& \& \& \& \& \& \& \\ 
\& \& \& \& $ $ \& $ $ \& |[fill=gray]| \& $ $ \& \& \& \& \& \& \& \& \& \\ 
\& \& \& \& $ $ \& $ $ \& $ $ \& |[fill=gray]| \& \& \& \& \& \& \& \& \& \\ 
\& \& \& \& |[fill=gray]| \& $ $ \& $ $ \& $ $ \& \& \& \& \& \& \& \& \& \\ 
\& \& \& \& \& \& \& \& $ $ \& $ $ \& |[fill=blue!20]| \& $ $ \& \& \& \& \\ 
\& \& \& \& \& \& \& \& $ $ \& $ $ \& $ $ \& |[fill=blue!20]| \& \& \& \& \\ 
\& \& \& \& \& \& \& \& |[fill=blue!20]| \& $ $ \& $ $ \& $ $ \& \& \& \& \\ 
\& \& \& \& \& \& \& \& $ $ \& |[fill=blue!20]| \& $ $ \& $ $ \& \& \& \& \\ 
\& \& \& \& \& \& \& \& \& \& \& \& $ $ \& $ $ \& $ $ \& |[fill=blue!80]| \\
\& \& \& \& \& \& \& \& \& \& \& \& |[fill=blue!80]| \& $ $ \& $ $ \& $ $ \\ 
\& \& \& \& \& \& \& \& \& \& \& \& $ $ \& |[fill=blue!80]| \& $ $ \& $ $ \\ 
\& \& \& \& \& \& \& \& \& \& \& \& $ $ \& $ $ \& |[fill=blue!80]| \& $ $ \\ 
};

\matrix [ampersand replacement=\&, right=1.5cm of V0] (V1) {
$ $ \& |[fill=gray]| \& $ $ \& $ $ \& \& \& \& \& \& \& \& \& \& \& \& \& \\ 
$ $ \& $ $ \& |[fill=gray]| \& $ $ \& \& \& \& \& \& \& \& \& \& \& \& \& \\ 
$ $ \& $ $ \& $ $ \& |[fill=gray]| \& \& \& \& \& \& \& \& \& \& \& \& \& \\ 
|[fill=gray]| \& $ $ \& $ $ \& $ $ \& \& \& \& \& \& \& \& \& \& \& \& \& \\ 
\& \& \& \& $ $ \& $ $ \& |[fill=blue!20]| \& $ $ \& \& \& \& \& \& \& \& \\ 
\& \& \& \& $ $ \& $ $ \& $ $ \& |[fill=blue!20]| \& \& \& \& \& \& \& \& \\ 
\& \& \& \& |[fill=blue!20]| \& $ $ \& $ $ \& $ $ \& \& \& \& \& \& \& \& \\ 
\& \& \& \& $ $ \& |[fill=blue!20]| \& $ $ \& $ $ \& \& \& \& \& \& \& \& \\ 
\& \& \& \& \& \& \& \& $ $ \& $ $ \& $ $ \& |[fill=blue!80]| \& \& \& \& \\
\& \& \& \& \& \& \& \& |[fill=blue!80]| \& $ $ \& $ $ \& $ $ \& \& \& \& \\ 
\& \& \& \& \& \& \& \& $ $ \& |[fill=blue!80]| \& $ $ \& $ $ \& \& \& \& \\ 
\& \& \& \& \& \& \& \& $ $ \& $ $ \&|[fill=blue!80]| \& $ $ \& \& \& \& \\ 
\& \& \& \& \& \& \& \& \& \& \& \& |[fill=blue!50]| \& $ $ \& $ $ \& $ $ \\ 
\& \& \& \& \& \& \& \& \& \& \& \& $ $ \& |[fill=blue!50]| \& $ $ \& $ $ \\ 
\& \& \& \& \& \& \& \& \& \& \& \& $ $ \& $ $ \& |[fill=blue!50]| \& $ $ \\ 
\& \& \& \& \& \& \& \& \& \& \& \& $ $ \& $ $ \& $ $ \& |[fill=blue!50]| \\ 
};

\matrix [ampersand replacement=\&, right=2.5cm of V1] (V3) {
$ $ \& $ $ \& $ $ \& |[fill=blue!80]| \& \& \& \& \& \& \& \& \& \& \& \& \\ 
|[fill=blue!80]| \& $ $ \& $ $ \& $ $ \& \& \& \& \& \& \& \& \& \& \& \& \\ 
$ $ \& |[fill=blue!80]| \& $ $ \& $ $ \& \& \& \& \& \& \& \& \& \& \& \& \\ 
$ $ \& $ $ \& |[fill=blue!80]| \& $ $ \& \& \& \& \& \& \& \& \& \& \& \& \\ 
\& \& \& \& |[fill=blue!50]| \& $ $ \& $ $ \& $ $ \& \& \& \& \& \& \& \& \\ 
\& \& \& \& $ $ \& |[fill=blue!50]| \& $ $ \& $ $ \& \& \& \& \& \& \& \& \\ 
\& \& \& \& $ $ \& $ $ \& |[fill=blue!50]| \& $ $ \& \& \& \& \& \& \& \& \\ 
\& \& \& \& $ $ \& $ $ \& $ $ \& |[fill=blue!50]| \& \& \& \& \& \& \& \& \\ 
\& \& \& \& \& \& \& \& $ $ \& |[fill=gray]| \& $ $ \& $ $ \\ 
\& \& \& \& \& \& \& \& $ $ \& $ $ \& |[fill=gray]| \& $ $ \\ 
\& \& \& \& \& \& \& \& $ $ \& $ $ \& $ $ \& |[fill=gray]| \\ 
\& \& \& \& \& \& \& \& |[fill=gray]| \& $ $ \& $ $ \& $ $ \\ 
\& \& \& \& \& \& \& \& \& \& \& \& $ $ \& $ $ \& |[fill=blue!20]| \& $ $ \\
\& \& \& \& \& \& \& \& \& \& \& \& $ $ \& $ $ \& $ $ \& |[fill=blue!20]| \\ 
\& \& \& \& \& \& \& \& \& \& \& \& |[fill=blue!20]| \& $ $ \& $ $ \& $ $ \\ 
\& \& \& \& \& \& \& \& \& \& \& \& $ $ \& |[fill=blue!20]| \& $ $ \& $ $ \\ 
};

\matrix [ampersand replacement=\&, below=0.5cm of V0] (W0) { 
|[fill=blue!50]| \& $ $ \& $ $ \& $ $ \& $ $ \& $ $ \& $ $ \& $ $ \& $ $ \& $ $ \& $ $ \& $ $ \& $ $ \& $ $ \& $ $ \& $ $ \\ 
$ $ \& $ $ \& $ $ \& $ $ \& $ $ \& |[fill=blue!50]| \& $ $ \& $ $ \& $ $ \& $ $ \& $ $ \& $ $ \& $ $ \& $ $ \& $ $ \& $ $ \\ 
$ $ \& $ $ \& $ $ \& $ $ \& $ $ \& $ $ \& $ $ \& $ $ \& $ $ \& $ $ \& |[fill=blue!50]| \& $ $ \& $ $ \& $ $ \& $ $ \& $ $ \\ 
$ $ \& $ $ \& $ $ \& $ $ \& $ $ \& $ $ \& $ $ \& $ $ \& $ $ \& $ $ \& $ $ \& $ $ \& $ $ \& $ $ \& $ $ \& |[fill=blue!50]| \\ 
$ $ \& $ $ \& $ $ \& $ $ \& $ $ \& |[fill=gray]| \& $ $ \& $ $ \& $ $ \& $ $ \& $ $ \& $ $ \& $ $ \& $ $ \& $ $ \& $ $ \\ 
$ $ \& $ $ \& $ $ \& $ $ \& $ $ \& $ $ \& $ $ \& $ $ \& $ $ \& $ $ \& |[fill=gray]| \& $ $ \& $ $ \& $ $ \& $ $ \& $ $ \\ 
$ $ \& $ $ \& $ $ \& $ $ \& $ $ \& $ $ \& $ $ \& $ $ \& $ $ \& $ $ \& $ $ \& $ $ \& $ $ \& $ $ \& $ $ \& |[fill=gray]| \\ 
|[fill=gray]| \& $ $ \& $ $ \& $ $ \& $ $ \& $ $ \& $ $ \& $ $ \& $ $ \& $ $ \& $ $ \& $ $ \& $ $ \& $ $ \& $ $ \& $ $ \\ 
$ $ \& $ $ \& $ $ \& $ $ \& $ $ \& $ $ \& $ $ \& $ $ \& $ $ \& $ $ \& |[fill=blue!20]| \& $ $ \& $ $ \& $ $ \& $ $ \& $ $ \\ 
$ $ \& $ $ \& $ $ \& $ $ \& $ $ \& $ $ \& $ $ \& $ $ \& $ $ \& $ $ \& $ $ \& $ $ \& $ $ \& $ $ \& $ $ \& |[fill=blue!20]| \\ 
|[fill=blue!20]| \& $ $ \& $ $ \& $ $ \& $ $ \& $ $ \& $ $ \& $ $ \& $ $ \& $ $ \& $ $ \& $ $ \& $ $ \& $ $ \& $ $ \& $ $ \\ 
$ $ \& $ $ \& $ $ \& $ $ \& $ $ \& |[fill=blue!20]| \& $ $ \& $ $ \& $ $ \& $ $ \& $ $ \& $ $ \& $ $ \& $ $ \& $ $ \& $ $ \\ 
$ $ \& $ $ \& $ $ \& $ $ \& $ $ \& $ $ \& $ $ \& $ $ \& $ $ \& $ $ \& $ $ \& $ $ \& $ $ \& $ $ \& $ $ \& |[fill=blue!80]| \\ 
|[fill=blue!80]| \& $ $ \& $ $ \& $ $ \& $ $ \& $ $ \& $ $ \& $ $ \& $ $ \& $ $ \& $ $ \& $ $ \& $ $ \& $ $ \& $ $ \& $ $ \\ 
$ $ \& $ $ \& $ $ \& $ $ \& $ $ \& |[fill=blue!80]| \& $ $ \& $ $ \& $ $ \& $ $ \& $ $ \& $ $ \& $ $ \& $ $ \& $ $ \& $ $ \\ 
$ $ \& $ $ \& $ $ \& $ $ \& $ $ \& $ $ \& $ $ \& $ $ \& $ $ \& $ $ \& |[fill=blue!80]| \& $ $ \& $ $ \& $ $ \& $ $ \& $ $ \\ 
};

\matrix [ampersand replacement=\&, right=1.5cm of W0] (W1) {
$ $ \& $ $ \& $ $ \& $ $ \& $ $ \& |[fill=blue!50]| \& $ $ \& $ $ \& $ $ \& $ $ \& $ $ \& $ $ \& $ $ \& $ $ \& $ $ \& $ $ \\ 
$ $ \& $ $ \& $ $ \& $ $ \& $ $ \& $ $ \& $ $ \& $ $ \& $ $ \& $ $ \& |[fill=blue!50]| \& $ $ \& $ $ \& $ $ \& $ $ \& $ $ \\ 
$ $ \& $ $ \& $ $ \& $ $ \& $ $ \& $ $ \& $ $ \& $ $ \& $ $ \& $ $ \& $ $ \& $ $ \& $ $ \& $ $ \& $ $ \& |[fill=blue!50]| \\ 
|[fill=blue!50]| \& $ $ \& $ $ \& $ $ \& $ $ \& $ $ \& $ $ \& $ $ \& $ $ \& $ $ \& $ $ \& $ $ \& $ $ \& $ $ \& $ $ \& $ $ \\ 
$ $ \& $ $ \& $ $ \& $ $ \& $ $ \& $ $ \& $ $ \& $ $ \& $ $ \& $ $ \& |[fill=gray]| \& $ $ \& $ $ \& $ $ \& $ $ \& $ $ \\ 
$ $ \& $ $ \& $ $ \& $ $ \& $ $ \& $ $ \& $ $ \& $ $ \& $ $ \& $ $ \& $ $ \& $ $ \& $ $ \& $ $ \& $ $ \& |[fill=gray]| \\ 
|[fill=gray]| \& $ $ \& $ $ \& $ $ \& $ $ \& $ $ \& $ $ \& $ $ \& $ $ \& $ $ \& $ $ \& $ $ \& $ $ \& $ $ \& $ $ \& $ $ \\ 
$ $ \& $ $ \& $ $ \& $ $ \& $ $ \& |[fill=gray]| \& $ $ \& $ $ \& $ $ \& $ $ \& $ $ \& $ $ \& $ $ \& $ $ \& $ $ \& $ $ \\ 
$ $ \& $ $ \& $ $ \& $ $ \& $ $ \& $ $ \& $ $ \& $ $ \& $ $ \& $ $ \& $ $ \& $ $ \& $ $ \& $ $ \& $ $ \& |[fill=blue!20]| \\ 
|[fill=blue!20]| \& $ $ \& $ $ \& $ $ \& $ $ \& $ $ \& $ $ \& $ $ \& $ $ \& $ $ \& $ $ \& $ $ \& $ $ \& $ $ \& $ $ \& $ $ \\ 
$ $ \& $ $ \& $ $ \& $ $ \& $ $ \& |[fill=blue!20]| \& $ $ \& $ $ \& $ $ \& $ $ \& $ $ \& $ $ \& $ $ \& $ $ \& $ $ \& $ $ \\ 
$ $ \& $ $ \& $ $ \& $ $ \& $ $ \& $ $ \& $ $ \& $ $ \& $ $ \& $ $ \& |[fill=blue!20]| \& $ $ \& $ $ \& $ $ \& $ $ \& $ $ \\ 
|[fill=blue!80]| \& $ $ \& $ $ \& $ $ \& $ $ \& $ $ \& $ $ \& $ $ \& $ $ \& $ $ \& $ $ \& $ $ \& $ $ \& $ $ \& $ $ \& $ $ \\ 
$ $ \& $ $ \& $ $ \& $ $ \& $ $ \& |[fill=blue!80]| \& $ $ \& $ $ \& $ $ \& $ $ \& $ $ \& $ $ \& $ $ \& $ $ \& $ $ \& $ $ \\ 
$ $ \& $ $ \& $ $ \& $ $ \& $ $ \& $ $ \& $ $ \& $ $ \& $ $ \& $ $ \& |[fill=blue!80]| \& $ $ \& $ $ \& $ $ \& $ $ \& $ $ \\ 
$ $ \& $ $ \& $ $ \& $ $ \& $ $ \& $ $ \& $ $ \& $ $ \& $ $ \& $ $ \& $ $ \& $ $ \& $ $ \& $ $ \& $ $ \& |[fill=blue!80]| \\ 
};

\matrix [ampersand replacement=\&, right=2.5cm of W1] (W3) {
$ $ \& $ $ \& $ $ \& $ $ \& $ $ \& $ $ \& $ $ \& $ $ \& $ $ \& $ $ \& $ $ \& $ $ \& $ $ \& $ $ \& $ $ \& |[fill=blue!50]| \\ 
|[fill=blue!50]| \& $ $ \& $ $ \& $ $ \& $ $ \& $ $ \& $ $ \& $ $ \& $ $ \& $ $ \& $ $ \& $ $ \& $ $ \& $ $ \& $ $ \& $ $ \\ 
$ $ \& $ $ \& $ $ \& $ $ \& $ $ \& |[fill=blue!50]| \& $ $ \& $ $ \& $ $ \& $ $ \& $ $ \& $ $ \& $ $ \& $ $ \& $ $ \& $ $ \\ 
$ $ \& $ $ \& $ $ \& $ $ \& $ $ \& $ $ \& $ $ \& $ $ \& $ $ \& $ $ \& |[fill=blue!50]| \& $ $ \& $ $ \& $ $ \& $ $ \& $ $ \\ 
|[fill=gray]| \& $ $ \& $ $ \& $ $ \& $ $ \& $ $ \& $ $ \& $ $ \& $ $ \& $ $ \& $ $ \& $ $ \& $ $ \& $ $ \& $ $ \& $ $ \\ 
$ $ \& $ $ \& $ $ \& $ $ \& $ $ \& |[fill=gray]| \& $ $ \& $ $ \& $ $ \& $ $ \& $ $ \& $ $ \& $ $ \& $ $ \& $ $ \& $ $ \\ 
$ $ \& $ $ \& $ $ \& $ $ \& $ $ \& $ $ \& $ $ \& $ $ \& $ $ \& $ $ \& |[fill=gray]| \& $ $ \& $ $ \& $ $ \& $ $ \& $ $ \\ 
$ $ \& $ $ \& $ $ \& $ $ \& $ $ \& $ $ \& $ $ \& $ $ \& $ $ \& $ $ \& $ $ \& $ $ \& $ $ \& $ $ \& $ $ \& |[fill=gray]| \\ 
$ $ \& $ $ \& $ $ \& $ $ \& $ $ \& |[fill=blue!20]| \& $ $ \& $ $ \& $ $ \& $ $ \& $ $ \& $ $ \& $ $ \& $ $ \& $ $ \& $ $ \\ 
$ $ \& $ $ \& $ $ \& $ $ \& $ $ \& $ $ \& $ $ \& $ $ \& $ $ \& $ $ \& |[fill=blue!20]| \& $ $ \& $ $ \& $ $ \& $ $ \& $ $ \\ 
$ $ \& $ $ \& $ $ \& $ $ \& $ $ \& $ $ \& $ $ \& $ $ \& $ $ \& $ $ \& $ $ \& $ $ \& $ $ \& $ $ \& $ $ \& |[fill=blue!20]| \\ 
|[fill=blue!20]| \& $ $ \& $ $ \& $ $ \& $ $ \& $ $ \& $ $ \& $ $ \& $ $ \& $ $ \& $ $ \& $ $ \& $ $ \& $ $ \& $ $ \& $ $ \\ 
$ $ \& $ $ \& $ $ \& $ $ \& $ $ \& $ $ \& $ $ \& $ $ \& $ $ \& $ $ \& |[fill=blue!80]| \& $ $ \& $ $ \& $ $ \& $ $ \& $ $ \\ 
$ $ \& $ $ \& $ $ \& $ $ \& $ $ \& $ $ \& $ $ \& $ $ \& $ $ \& $ $ \& $ $ \& $ $ \& $ $ \& $ $ \& $ $ \& |[fill=blue!80]| \\ 
|[fill=blue!80]| \& $ $ \& $ $ \& $ $ \& $ $ \& $ $ \& $ $ \& $ $ \& $ $ \& $ $ \& $ $ \& $ $ \& $ $ \& $ $ \& $ $ \& $ $ \\ 
$ $ \& $ $ \& $ $ \& $ $ \& $ $ \& |[fill=blue!80]| \& $ $ \& $ $ \& $ $ \& $ $ \& $ $ \& $ $ \& $ $ \& $ $ \& $ $ \& $ $ \\ 
};

\node[left=0.5cm of V0] {$V_0$};
\node[left=0.125cm of V0] {$=$};
\node[left=0.5cm of V1] {,$V_1$};
\node[left=0.125cm of V1] {$=$};

\node[left=0.5cm of V3] {, $\ldots$, $V_3$};
\node[left=0.125cm of V3] {$=$};

\node[left=0.5cm of W0] {$W_0$};
\node[left=0.125cm of W0] {$=$};
\node[left=0.5cm of W1] {,$W_1$};
\node[left=0.125cm of W1] {$=$};

\node[left=0.5cm of W3] {, $\ldots$, $W_3$};
\node[left=0.125cm of W3] {$=$};
\end{tikzpicture}
}
\end{center}
\caption{Matrix multiplication - $V_k$ and $W_k$ examples for $d=4$}
\end{figure*}
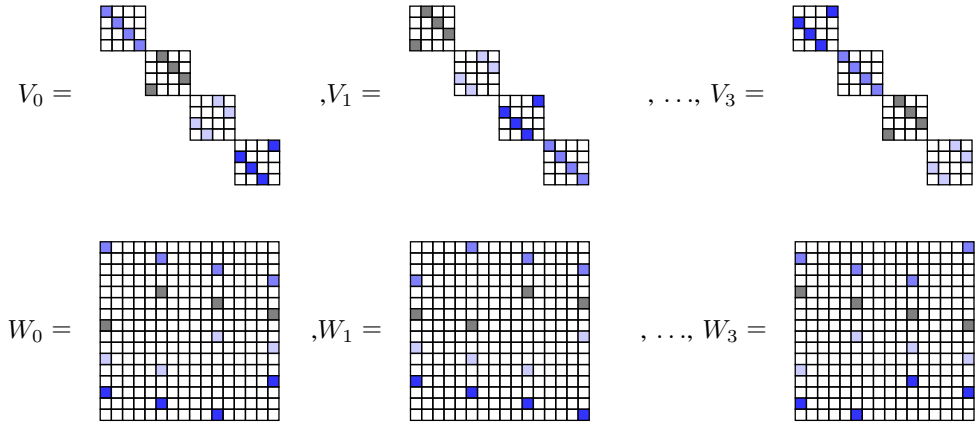

\begin{table}[h!]
  \begin{tabular}{C{1.75cm} | C{1.25cm} | C{1.25cm} | C{1cm} | C{1.25cm}}
    Methodology & No. of Ciphertexts & Complexity & Mult. depth & Relineariza-tions\\
    \hline
    Halevi-Shoup & $d$ & $\mathcal{O}(d^2)$ & $1 \bullet$ & $1$ \\
    JKLS & $1$ & $\mathcal{O}(d)$ & $1 \bullet + 2 \odot$ & $3$ \\
    Our work & $1$ & $\mathcal{O}(d)$ & $1 \bullet + 1 \odot$ & $2$ \\
  \end{tabular}
  \caption{Comparison of the different matrix multiplication algorithms.}
\label{tbl:HaleviShoup}
\end{table}

The matrices $V_k$ and $W_k$ are permuting the row-encoded matrices $A$ and $B$ respectively such that the matrix multiplication algorithm as we know can be implemented with element-wise multiplication and additions. Table \ref{tbl:HaleviShoup} summarizes the complexity differences of our method compared to Halevi-Shoup and JKLS and Algorithm \ref{algo:MMult} the implementation of our methodology. With this reduction of 1 $\odot$ operation in the matrix the multiplication algorithm we are able to perform 9 and 6 iterations on GD and AGD respectively, as opposed to 6 and 4 iterations if we were using the JKKS multiplication scheme.

\begin{algorithm}[ht]
\caption{HE Matrix Multiplication (\textsc{MMult})}
\label{algo:MMult}
\begin{algorithmic}
  \Function{MMult}{$\mathtt{A}, \mathtt{B}, d, a$}
  \State $\mathtt{AB} \gets \textsc{CipherText()}$
  \For {$k = 0$ to $d-1$}
    \State $\mathtt{A_k} \gets \textsc{LinTrans}(\mathtt{A_0}, V_k(a))$
    \State $\mathtt{B_k} \gets \textsc{LinTrans}(\mathtt{B_0}, W_k(1))$
    \State $\mathtt{AB_k} \gets \mathtt{A_k} \bullet \mathtt{B_k}$
    \State $\textsc{Relinearize}(\mathtt{AB_k})$
    \State $\mathtt{AB} \gets \mathtt{AB} \Plus \mathtt{AB_k}$
  \EndFor
  \State \textbf{return} $\mathtt{AB}$
  \EndFunction
\end{algorithmic}
\end{algorithm}

\subsection{Extension to other QP problems}

Extension to other QP problems is feasible. For instance, linear equality constraints could be handled by converting the problem to an unconstrained QP, or by solving primal-dual methods. These approaches sound completely viable but are subject to the multiplication depth limitations on the HE circuit. In other words, extra arithmetic operations can take place, but at the cost of reducing the number of maximum iterations. Linear inequality constraints are not directly supported, but an approach would be to decrypt and re-encrypt at each iteration (not actually a practical solution).

%% file: chapters/numerical.tex
\section{Numerical analysis} \label{sec:SecV}

The HE resource requirements are directly proportional to the capacity of the encryption circuit. The larger the circuit's capacity, the larger the computing memory and computational power required at each arithmetic operation. Given our computing resources and the parameters of the Microsoft SEAL \cite{SEAL21}, the largest circuit we can implement, in the CKKS scheme, has a multiplication depth of 18. At each iteration, AGD and GD have a multiplication depth of 3 and 2, resulting in a cap of 6 and 9 steps for AGD and GD respectively. Even though we could not implement longer iterations due to limitations on our computational resources, we believe that the results would apply in that case too.


We start by running both algorithms with initial condition $x_0 \neq x^*$ for the same matrix $Q$ and decrypt the outcomes at every iteration. Figure \ref{fig:gd_steps} shows that at each iteration the solution gets closer to the optimal $x^*$.
\begin{figure}[h!]
  \centering
  \includegraphics[width=0.475\textwidth]{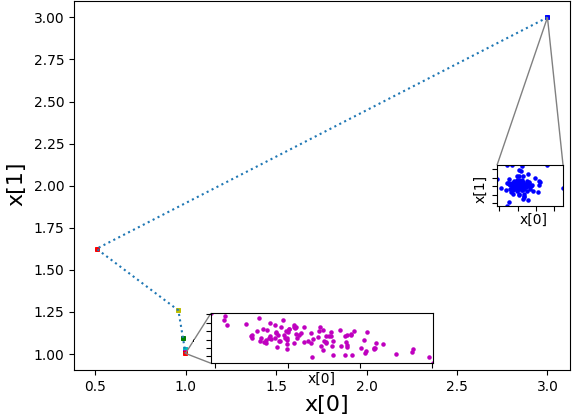}
  \caption{Decrypted HE-AGD steps for 100 repetitions with a 2-by-2 matrix with $\kappa=2$, optimal value at $x^* = (1,1)$ initial condition $x_0=(3,3)$. A similar behavior is observed for HE-GD.}
  \label{fig:gd_steps}
\end{figure}

\begin{figure*}[ht!]
  \centering
  \includegraphics[width=0.45\textwidth]{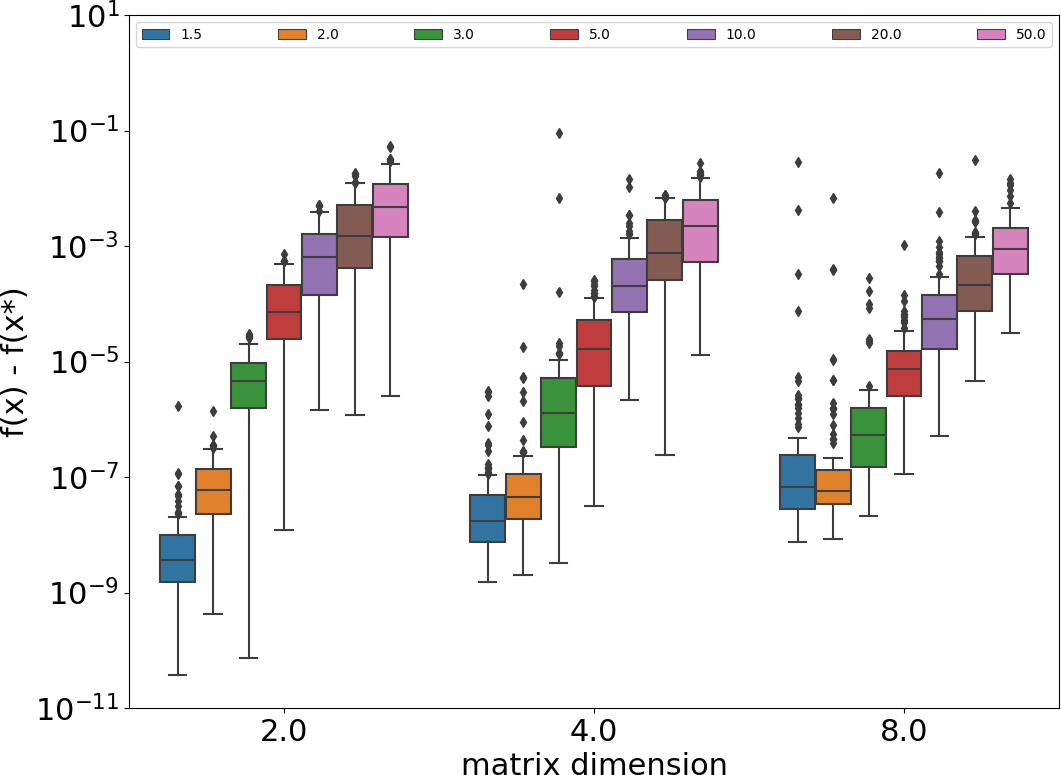}
  \includegraphics[width=0.45\textwidth]{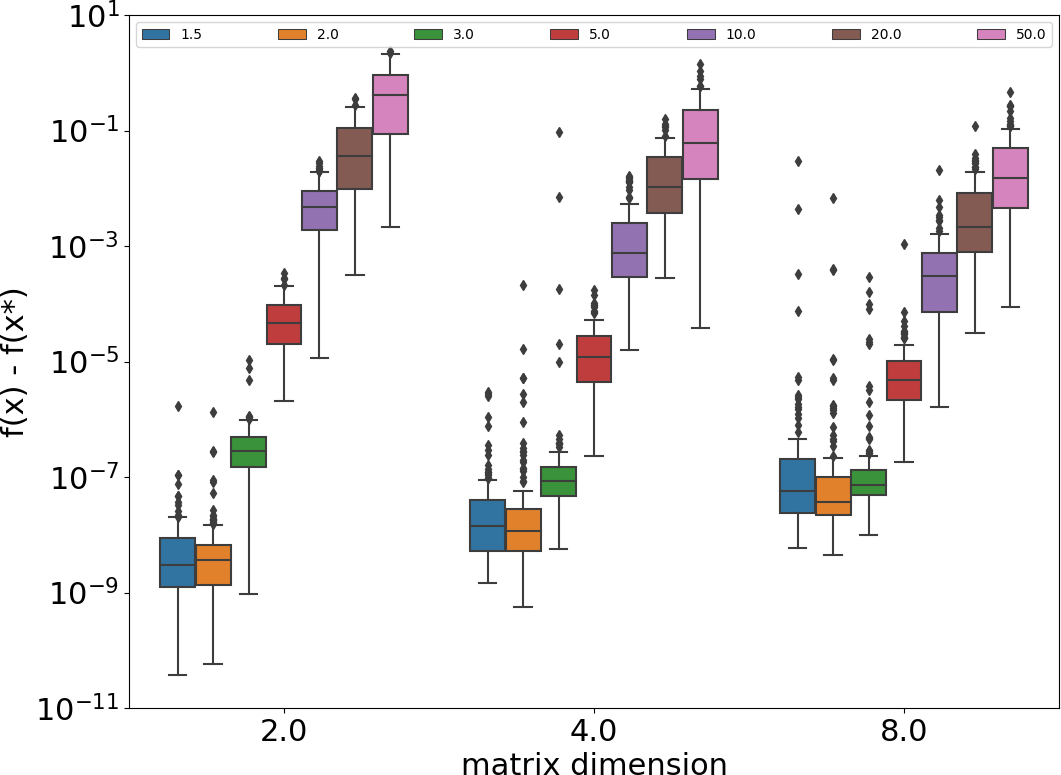}
  \caption{Box plot of tolerance $f(x) - f(x^*)$ for AGD (left) and GD (right) with matrices of different dimensions (x axis) and different $\kappa$ (colored bars).}
  \label{fig:agd_gd_comparison}
\end{figure*}

As discussed in Section \ref{sec:SecII} AGD exhibits a superior convergence rate, hence it allows meeting a given convergence (in terms of optimality) tolerance with fewer iterations. However, in case of encryption, the computational limits imposed by the allowable depth of the encryption circuit introduces a trade-off, as AGD involves more arithmetic operations compared to GD (see Section IV-A). As such, it might be computationally impossible to perform the number of iterations needed by AGD to meet a given tolerance. We investigate this trade-off numerically, and show that the preferred method depends on the condition number $\kappa$ of the quadratic matrix $Q$.

To analyze this trade-off numerically we generate QP instances with condition number $\kappa$ ranging from 1.5 to 50. To collect numerical statistics on the effect of $\kappa$, for each $\kappa$ we generate 100 sets of randomly generated symmetric positive-definite matrices $Q$ of dimension 2, 4 and 8\footnote{Higher dimensions are also feasible, and this is independent of the multiplication depth limits.}, and associated random vectors $p$. We solve each QP instance via the HE-GD and HE-AGD methods with an initial condition $x_0$ such that $\norm{x_0 - x^*}_2 = 1$. Our goal is to investigate which algorithm achieves better tolerance values (in terms distance to the optimal value) at the last iteration allowed by the encryption's circuit depth. The latter is iteration 6 for AGD and iteration 9 for GD. 
Figure \ref{fig:agd_gd_comparison} illustrates the distribution of the tolerance $f(x) - f(x^*)$ (optimality gap of the returned solution $x$ from the optimal cost $f(x^*$) for AGD (top) and GD (bottom) with matrices of different dimensions (x axis) and different values of the condition number $\kappa$ (color code).

Table \ref{tbl:agd_gd_comparison} highlights the important observations stemming from Figure \ref{fig:agd_gd_comparison}. In particular, GD profits from the extra iterations and achieves better tolerance (getting closer to the optimal) values when $\kappa \leq 5$ (upper table). Yet, AGD outperforms GD in cases where $\kappa>5$ (lower table) although with worse tolerance values (that is, further form the optimal). This shows that the limitation on multiplication depth is a major issue when handling matrices with high values of $\kappa$. 

\begin{table}[h!]
\centering
  \begin{tabular}{c|c|c|c|c}
     \diagbox{$d$}{$\kappa$}& 1.5 & 2 & 3 & 5 \\
    \hline
    2 & \cellcolor{blue!25} $3 \cdot 10^{-9}$ & \cellcolor{blue!25} $4 \cdot 10^{-9}$ & \cellcolor{blue!25} $3\cdot10^{-7}$ & \cellcolor{blue!25} $5\cdot10^{-5}$ \\
    4 & \cellcolor{blue!25} $1 \cdot 10^{-8}$ & \cellcolor{blue!25} $1 \cdot 10^{-8}$ & \cellcolor{blue!25} $8\cdot10^{-8}$ & \cellcolor{blue!25} $1\cdot10^{-5}$ \\
    8 & \cellcolor{blue!25} $6 \cdot 10^{-8}$ & \cellcolor{blue!25} $4 \cdot 10^{-8}$ & \cellcolor{blue!25} $7\cdot10^{-8}$ & \cellcolor{blue!25} $5\cdot10^{-6}$
    \end{tabular}

  \begin{tabular}{c|c|c|c}
     \diagbox{$d$}{$\kappa$} & 10 & 20 & 50 \\
    \hline
    2 & \cellcolor{green!25} $7\cdot10^{-3}$ & \cellcolor{green!25} $2\cdot10^{-3}$ & \cellcolor{green!25} $5\cdot10^{-3}$\\
    4 & \cellcolor{green!25} $2\cdot10^{-4}$ & \cellcolor{green!25} $8\cdot10^{-4}$ & \cellcolor{green!25} $2\cdot10^{-3}$\\
    8 & \cellcolor{green!25} $6\cdot10^{-5}$ & \cellcolor{green!25} $2\cdot10^{-4}$ & \cellcolor{green!25} $9\cdot10^{-4}$ 
    \end{tabular}
  \caption{Comparison of AGD (6th iteration) against GD (9th iteration) for matrices of different sizes and conditional values $\kappa$. The upper and lower table indicate the range of values for $\kappa$ for which the GD and the AGD algorithm, respectively, are preferable in terms of returning a solution closer to the optimal one. The numbers represent the median tolerance level $f(x) - f(x^*)$ out of the 100 repetitions corresponding to different matrices $Q$.}
\label{tbl:agd_gd_comparison}
\end{table}

The code running the numerical examples presented here (https://github.com/f2cf2e10/agd-he) used our own Python wapper of the Microsoft SEAL \cite{SEAL21} C++ library (https://github.com/f2cf2e10/pSEAL). We used an Intel Xeon E5-1620 with 24GB of RAM machine running Debian 11.

%% file: chapters/conclusion.tex
\section{Conclusion} \label{sec:SecVI}

In this paper we studied, implemented, and analyzed both gradient and accelerated gradient descent algorithms to solve QP problems in a HE fashion. We evaluated different encryption schemes (BFV, BGV, and CKKS) and argued that CKKS is the only suitable scheme for implementing gradient descent methods as it allows for freedom in the choice of the algorithm's step size. In our implementation, AGD takes an extra multiplication operation at each step when compared to GD. As a result, AGD cannot run for as many steps as GD for an encryption circuit with the same security parameters and channel capacity (i.e. multiplication depth). We demonstrate that the condition number of the matrix of the quadratic term of the objective function plays an important role in determining which algorithm is preferred. For higher values of the condition number, AGD is preferred, as it converges faster to the solution, even if performing fewer iterations. Whereas for lower values of the condition number of the matrix GD performs better due to the extra iterations. In addition, we proposed a new HE matrix multiplication algorithm that is more efficient, in terms of multiplication depth, than other algorithms proposed in the literature. These observations have been verified by means of a numerical investigation. Yet, there are still many outstanding challenges in HE version of iterative numerical procedures. For instance, solving for constrained problems is in general a challenge, given the extra operations required to project into the constrained set. To solve these issues, further work should focus on optimizing for the multiplication depth and perhaps explore alternative ways to encode matrices in the HE scheme.